\documentclass[11pt,reqno]{amsart}
\usepackage{amssymb}
\usepackage{comment}
\usepackage{graphicx}

\usepackage[margin=1in]{geometry}
\usepackage[pagebackref, colorlinks = true, linkcolor = blue, urlcolor  = blue, citecolor = red]{hyperref}

\newtheorem{theorem}{Theorem}[section]
\newtheorem{definition}[theorem]{Definition}
\newtheorem{proposition}[theorem]{Proposition}
\newtheorem{corollary}[theorem]{Corollary}
\newtheorem{lemma}[theorem]{Lemma}
\newtheorem{remark}[theorem]{Remark}

\newcommand{\id}{\mathrm{id}}
\newcommand{\tr}{\mathrm{tr}}
\newcommand{\Tr}{\mathrm{Tr}}
\newcommand{\I}{\mathrm{I}}

\newcommand{\CC}{\mathbb{C}}

\newcommand{\Wg}{\mathrm{Wg}}
\newcommand{\E}{\mathbb{E}}
\newcommand{\Moeb}{\mathrm{M{\ddot o}b}}

\renewcommand{\epsilon}{\varepsilon}

\begin{document}

\title{On the spectral gap of random quantum channels}

\author{Carlos E. Gonz\'alez-Guill\'en}
\email{carlos.gguillen@upm.es}
\address{Departamento de Matem\'atica Aplicada a la Ingenier\'ia Industrial, Universidad Polit\'ecnica de Madrid, Madrid, Spain, \& IMI, Universidad Complutense de Madrid, Madrid, Spain.}

\author{Marius Junge}
\email{junge@math.uiuc.edu}
\address{Department of Mathematics, University of Illinois, Urbana, IL 61801, USA}

\author{Ion Nechita}
\email{nechita@irsamc.ups-tlse.fr}
\address{Laboratoire de Physique Th\'eorique, Universit\'e de Toulouse, CNRS, UPS, France}

\subjclass[2000]{}

\date{\today}

\begin{abstract}
In this work, we prove a lower bound on the difference between the first and second singular values of quantum channels induced by random isometries, that is tight in the scaling of the number of Kraus operators. This allows us to give an upper bound on the difference between the first and second largest (in modulus) eigenvalues of random channels with same large input and output dimensions for finite number of Kraus operators $k\geq 169$. Moreover, we show that these random quantum channels are quantum expanders, answering a question posed in \cite{hastings2007random}. As an 
application, we show that ground states of infinite 1D spin chains, which are well-approximated by matrix product states, fulfill a principle of maximum entropy.
\end{abstract}

\maketitle

\tableofcontents

\section{Introduction}

Quantum channels are the most general linear transformations quantum systems can undergo. As such, the study of their properties is of central importance in Quantum Information Theory \cite{watrous2018theory}, much as Markov chains are central objects in probability theory and Shannon theory. 

During the last 15 years, there has been an increasing interest in the study of \emph{generic quantum channels}, that is random quantum channels having natural probability distributions. Such an approach, of paramount importance in classical information theory (where, e.g., random channel coding plays a central role) has been pursued in the field of quantum information, in parallel with the study of other generic objects, such as density matrices \cite{collins2016random}. 

Several models of random quantum channels have been considered in the literature. Channels where the Kraus operators are taken to be random unitary matrices were first considered in \cite{hayden2004randomizing} where it was shown (and later improve in \cite{aubrun2009almost}) that they are $\epsilon$-randomizing maps. The same channels were considered in \cite{hayden2005multiparty} as part of a data hiding scheme, in \cite{harrow2004superdense} for superdense coding, and \cite{dupuis2013locking} to study locking classical information. All this application have in common that the number of Kraus operators is increasing with the dimension of the system. However, in the case when  the number of Kraus operators is kept fixed, Hastings showed that random quantum channels are quantum expanders \cite{hastings2007random}, providing the quantum counterpart of Friedman's classical result that random regular graphs are (nearly) Ramanujan expanders \cite{friedman2008proof}. Later, Hastings used the same model of random channels to give the first counterexamples of the additivity of the minimum output entropy \cite{hastings2009superadditivity}.

Another model of random quantum channels stems from the Stinespring dilation theorem. Indeed, choosing a Haar-distributed random isometry induces a probability distribution on the set of quantum channels. The main difference from this model and the previous one is that in the random isometry case, channels are no longer unital (with probability one). Such models were first considered in \cite{hayden2008counterexamples} to tackle additivity related questions, and have been shown to also provide additivity counterexamples in \cite{fukuda2010entanglement}.

\medskip

In this work, we study the singular value and spectral gap of random quantum channels induced by random isometries which have a fixed number of Kraus operators, but very large input and output dimensions. Besides the fundamental importance of answering such questions, we motivate our work with two applications. First, we show that random quantum channels are quantum expanders, i.e.~they have a lower bounded spectral (eigenvalue) gap for fixed number of Kraus operators. Our results are the counterpart of Hastings' first example of such random channels \cite{hastings2007random}, in the framework of random mixed unitary channels. Our bounds on the gap match (up to constants) Hastings' bounds, as well as bounds obtained by Pisier in \cite{pisier2014quantum}. To be more precise, we show that for the number of Kraus operators $k\geq 169$ and input and output dimensions equal and large, random quantum channels $\Phi_n$ are generalized quantum expanders, where the invariant state does not need to be maximally mixed, but an state close to it, having a large von Neumann entropy. As a second application of our results, we prove that reduced density matrices of infinite translationally-invariant matrix product states satisfy a maximum entropy principle.

\medskip

The paper is organized as follows. In Section \ref{sec:random-quantum-channels}, we introduce in detail the model of random quantum channels we consider and the precise asymptotical regime in which we study them. Section \ref{sec:Weingarten} contains a brief review of Weingarten calculus needed for random matrix computations. Sections \ref{sec:LB} and \ref{sec:UB} are the main technical core of the paper, containing lower and, respectively, upper bounds on the singular values of the sequence of super-operators corresponding to the random quantum channels under consideration; these results are put together in  Section \ref{sec:spectral-gap} to establish the singular value and spectral gap of random quantum channels. In Section \ref{sec:quantum-expander} we establish that random quantum channels are quantum expanders, after proving that the Perron-Frobenius eigenvector has large entropy. Finally, in Section \ref{sec:TIMPS} we discuss our second application, a maximum entropy principle for infinite translationally-invariant matrix product states. 

\section{Random quantum channels}\label{sec:random-quantum-channels}

In this section we introduce the model of random quantum channels we shall study in the paper. Introduced in \cite{hayden2008counterexamples} in order to tackle the additivity problem for the minimum output entropy, the model of quantum channels we consider has received a lot of attention in the recent years, mainly due to its generality and simplicity. Indeed, while other models of randomness (e.g.~random mixed unitary channels, see \cite{hayden2004randomizing} or \cite{hastings2007random}) have also been considered, the model we describe below is the most natural from a probabilistic perspective, since it does not impose any constraints on the linear map, except complete positivity and trace preservation.

Let us fix a triple of integers $(d,n,k)$ satisfying $d \leq nk$ and consider an isometry
 $V:\mathbb C^d \to \mathbb C^n \otimes \mathbb C^k$; the isometric property reads $V^*V - I_d$. We shall choose the isometry $V$ at random, from the \emph{Haar measure}. Indeed, there is a unique probability measure on the set of isometries $\CC^d \to \CC^{nk}$ which is invariant under left and right multiplication by arbitrary unitary operators of appropriate size. In practice, one can sample such a random Haar isometry by truncating $nk-d$ columns off a random Haar unitary matrix $U$ of size $nk$. 
 
A random isometry as above induces a \emph{random quantum channel} 
\begin{align*}
    \Phi:M_d(\mathbb C) &\to M_n(\mathbb C)\\
X &\mapsto [\mathrm{id}_n \otimes \mathrm{Tr}_k](VXV^*).
\end{align*}

To a (random) quantum channel $\Phi$ as above, we associate three important objects (see \cite[Section 2.2]{watrous2018theory} for the theory on the subject):
\begin{itemize}
    \item its \emph{Kraus decomposition} $\Phi(X) = \sum_{i=1}^k A_i X A_i^*$. The matrices $A_i$ are actually the $n \times d$ blocks of the isometry $V$:
    $$V = \sum_{i=1}^k A_i \otimes e_i,$$
    for some fixed orthonormal basis $\{e_i\}_{i=1}^k$ of $\CC^k$. These are usually refer as Kraus operators of the channel $\Phi$.
    \item its \emph{Choi matrix} $C_\Phi \in M_{nd}(\CC)$, defined as the action of the channel $\Phi$ on half of a maximally entangled state:
    $$C_\Phi := [\Phi \otimes \mathrm{id}](\omega_d).$$
    Above, $\omega_d := \Omega_d\Omega_d^*$ is the rank-one projection on the maximally entangled state
    \begin{equation}\label{eq:maximally-entangled-vector}
        \CC^d \otimes \CC^d \ni \Omega_d := \frac{1}{\sqrt d} \sum_{i=1}^d e_i \otimes e_i
    \end{equation}
    for some fixed basis $\{e_i\}_{i=1}^d$ of $\CC^d$.
    \item its \emph{super-operator} $F\in M_{n^2 \times d^2}(\mathbb C)$, which is the matrix of $\Phi$, seen as a linear operator $\Phi:\CC^{d^2} \to \CC^{n^2}$. It's an easy exercise to show that 
    \begin{equation}\label{eq:super-operator}
    F = \sum_{i= 1}^k A_i \otimes \overline{A_i}.
    \end{equation}
    Moreover, $F$ is the \emph{realignment} of the Choi matrix $C_\Phi$, see \cite[Section 10.2]{bengtsson2006geometry}.
\end{itemize}
The Choi matrix and the super-operator are depicted in Figure \ref{fig:quantum-channel} in the Penrose tensor notation. We recall that in this notation, tensors are depicted by boxes, and tensor contractions, such as traces or matrix multiplications, are depicted by wires connecting different decorations attached to the boxes. The shape of the decorations indicate to which vector space they correspond, different shapes corresponding to vector spaces of different dimensions.

\begin{figure}[ht]
\centering
\includegraphics[scale=1]{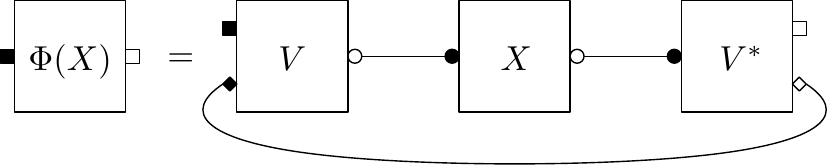}\\
\vspace{.5cm}
\includegraphics[scale=1]{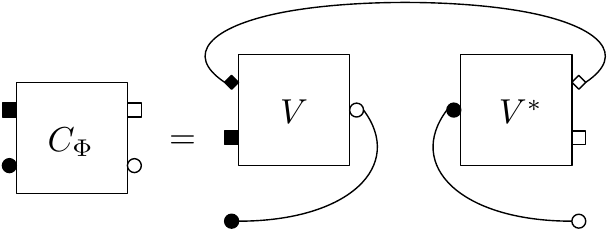} \qquad\qquad\qquad\qquad \includegraphics[scale=1]{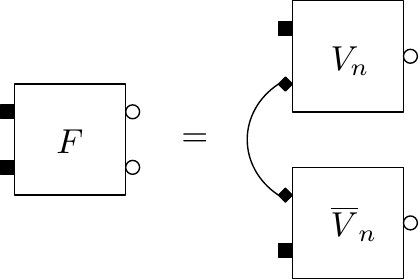}
\caption{A graphical representation of the output $\Phi(X)$ of a quantum channel given by an isometry $V$. In the bottom panel, the Choi matrix (left) and the super-operator (right) associated to $\Phi$. Round shaped decorations correspond to $\CC^d$, square shaped decorations correspond to $\CC^n$, and diamond shaped decorations correspond to $\CC^k$.}
\label{fig:quantum-channel}
\end{figure}

In what follows we shall study sequences of random quantum channels of increasing size, in the following \emph{asymptotic regime}. We shall assume that the parameter $k$ (the size of the environment, or the number of Kraus operators) is a fixed positive integer, and we shall consider a sequence of integers $d_n$ which will behave like $d_n \sim \lambda n$, for another fixed real constant $\lambda \in (0,k)$. We shall then define a sequence of random quantum channels $\Phi_n : M_{d_n}(\CC) \to M_n(\CC)$ as above, starting from Haar distributed random isometries $V_n : \CC^{d_n} \to \CC^n \otimes \CC^k$. We summarize our hypotheses below:
\begin{equation}\label{eq:asymptotic-regime}
    \begin{cases}
    & k \geq 1 \text{ fixed integer}\\
    & n \to \infty\\
    & \lambda \in (0,k) \text{ fixed real}\\
    & d_n \to \infty, \, d_n \sim \lambda n\\
    & V_n : \CC^{d_n} \to \CC^n \otimes \CC^k \text{ Haar-distributed random isometry}\\
    & \Phi_n : M_{d_n}(\CC) \to M_n(\CC) \text{ random quantum channel induced by } V_n.
    \end{cases}
\end{equation}

\section{Integration over the unitary group. The Weingarten function.}\label{sec:Weingarten}

Here we describe the main ingredients of the unitary Weingarten function and calculus that we are going to use to compute averages over the unitary group; for a complete description of this function we refer to \cite{collins2003moments,collins2006integration}. 

Let us start with some notation needed in the combinatorial study of permutations. For a permutation $\sigma \in \mathcal S_p$, we denote by $|\sigma|$ its \emph{length}, that is the minimum number $k$ such that $\sigma$ can be written as a product of $k$ transpositions. We shall write $\#\sigma$ for the \emph{number of cycles} of $\sigma$ (including the trivial fixed points). These quantities are related by the formula $|\sigma| + \#\sigma = p$.

\begin{definition}
The unitary Weingarten function $\Wg(n,\sigma)$ is a function taking as inputs a dimension parameter $n$ and a permutation $\sigma$ in the symmetric group $\mathcal S_p$. It is the pseudo inverse of the function $\sigma\mapsto n^{\# \sigma}$ under the convolution for the symmetric group.
\end{definition}

The interest of the Weingarten function lies in the following theorem from \cite{collins2003moments}, which states that the average of a monomial over the unitary group can be computed in terms of sums of Weingarten functions. We shall use the notation $[n]:=\{1,2,\ldots, n\}$.

\begin{theorem}
Let n be a positive integer and $i=(i_1,...,i_p)$, $i'=(i'_1,...,i'_p)$, $j=(j_1,...,j_p)$ and $j'=(j'_1,...,j'_p)$ be $p$-tuples of positive integers from $[n]$. Then

\begin{equation}\label{eq:Weingarten-formula}
\int_{\mathcal U_n} U_{i_1j_1} \cdots U_{i_pj_p} \bar U_{i'_1j'_1} \cdots \bar U_{i'_pj'_p}\, \mathrm{d}U=\sum_{\sigma,\tau \in \mathcal S_p} \delta_{i_1i'_{\sigma(1)}}... \delta_{i_pi'_{\sigma(p)}}\delta_{j_1j'_{\tau(1)}} ... \delta_{j_pj'_{\tau(p)}} \Wg(n,\sigma^{-1}\tau).
\end{equation}
\end{theorem}

In \cite{collins2010random} the authors introduce a graphical paradigm in order to simplify the use of the above formula in practice. Suppose that one wants to compute the expected value of a polynomial quantity defined in terms of a Haar-distributed random unitary matrix $U \in \mathcal U_n$. Assume also that this quantity is given as a diagram $\mathcal D$ in the Penrose tensor notation. Then, 
\begin{equation}\label{eq:Wg-graphical}
\E_U \mathcal D=\sum_{\sigma,\tau \in \mathcal S_p } C_{\sigma,\tau} \Wg(n,\sigma^{-1}\tau),
\end{equation}
where the diagrams $C_{\sigma, \tau}$ can be computed by the procedure depicted in Figure \ref{fig:Wg-graphical}. One has to enumerate the matrices $U$ and $\bar U$ appearing in $\mathcal D$ from $1$ to $p$; if the numbers of $U$ and $\bar U$ boxes in $\mathcal D$ are not identical, then the result is zero. 
For any pair of permutations $\sigma, \tau \in \mathcal S_p$, we construct the new diagram $\mathcal C_{\sigma, \tau}$ as follows: we delete the $U$ and $\bar U$ boxes and we connect the inputs of $U_i$ with the inputs of $\bar U_{\sigma(i)}$, and analogously, we connect the outputs using the $\tau$ permutation. In the resulting diagram $C_{\sigma, \tau}$ loops represent traces over the identity on some Hilbert space, and should be replaced by the appropriate dimension.

\begin{figure}[ht]
\centering
\includegraphics[scale=1]{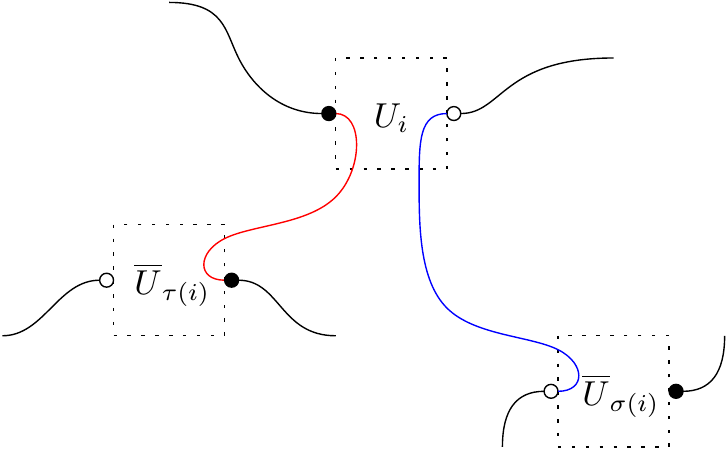}
\caption{A graphical interpretation of the Weingarten formula: in a diagrammatic notation, the boxes $U$ and $\bar U$ are deleted, and extra wires are used to connect the inputs (blue) according to $\sigma$ and the outputs (red) according to $\tau$.}
\label{fig:Wg-graphical}
\end{figure}

The following result from \cite{collins2003moments} gives an asymptotic bound of the Weingarten function when the size $n$ of the unitary matrices is growing, while the order $p$ of the monomial integrand is fixed.

\begin{lemma}\label{lem:Wg-asymptotic}
If $p \geq 1$ is a fixed integer and $\sigma\in \mathcal S_p$ is any permutation, then, as $n \to \infty$,
\begin{equation}\label{eq:Wg-asymptotic}
\Wg(n,\sigma)=n^{-p-|\sigma|}\Moeb(\sigma)(1+O(n^{-2})),
\end{equation}
where $\Moeb(\sigma)$ is a function which is multiplicative on the cycles of $\sigma$; its value for a full $p$-cycle is 
$$\Moeb((1,2,\cdots, p)) = (-1)^{p-1}\mathrm{Cat}_{p-1},$$
where $\mathrm{Cat}_p$ is the $p$-th Catalan number
$$\mathrm{Cat}_p = \frac{1}{p+1}\binom{2p}{p}.$$
\end{lemma}

\section{Lower bound on the norm}\label{sec:LB}
As described in the introduction, our strategy for showing that random quantum channels have a large spectral gap consists of two steps. In this section we accomplish the first step, providing a lower bound on the largest singular value (i.e.~the operator norm) of the super-operator $F$. 

\begin{proposition}\label{prop:limit-overlap}
	Consider a sequence of random quantum channels $\Phi_n:M_{d_n}(\CC) \to M_n(\CC)$ as in \eqref{eq:asymptotic-regime} and let $F_n$ be the corresponding super-operators \eqref{eq:super-operator} associated with the channel $\Phi_n$. Define the overlap
	$$\mathbb R \ni f_n:= \Omega_{d_n}^* F_n^* F_n \Omega_{d_n} = \Tr[\omega_{d_n} \cdot F_n^*F_n],$$
	where $\Omega_{d_n}$ is the maximally entangled vector \eqref{eq:maximally-entangled-vector} on the input space $\CC^{d_n}$ and $\omega_{d_n}$ is the corresponding quantum state.
	Then, for all integers $p \geq 1$

\begin{equation}
\lim_{n \to \infty} \mathbb E f_n^p = \left(\lambda+\frac 1 k -\frac \lambda {k^2}\right)^{p}.
\end{equation}
\end{proposition}

\begin{figure}[ht]
\centering
\includegraphics[scale=1]{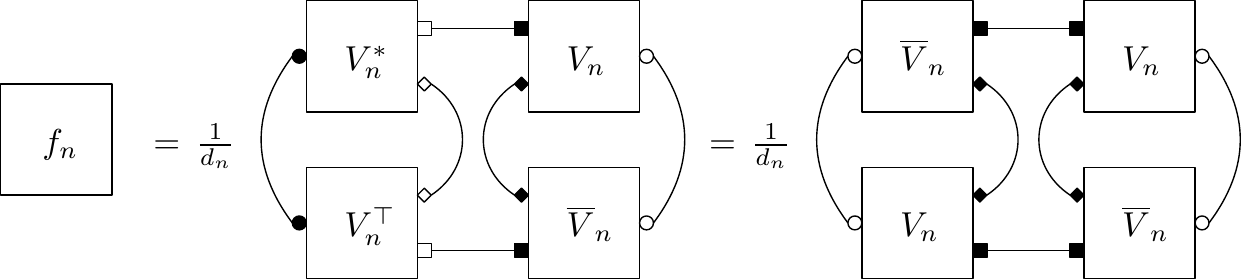}
\caption{Graphical representation of the random (scalar) overlap $f_n$. The round decorations correspond to the Hilbert space $\CC^{d_n}$, the square decorations correspond to $\CC^n$, and the diamond decorations correspond to $\CC^k$. The normalization factor $d_n^{-1}$ comes from the vectors $\Omega_{d_n}$.}
\label{fig:f_n}
\end{figure}

\begin{proof}
The proof is a straightforward application of the Weingarten graphical calculus. We shall compute the moments of the random variable $f_n$. For a moment of order $p$, we need to evaluate the expectation value of a diagram consisting of $p$ disjoint copies of the diagram in Figure \ref{fig:f_n}. Such a diagram contains $2p$ copies of $V_n$ (or $V_n^\top$) boxes, and $2p$ copies of $\overline V_n$ (or $V_n^*$) boxes, so the sum in the Weingarten formula \eqref{eq:Weingarten-formula} is indexed by two permutations of $2p$ elements:
$$\E f_n^p =\sum_{\alpha,\beta \in \mathcal S_{2p} } C_{\alpha,\beta} \Wg(nk,\alpha^{-1}\beta),$$
where the $C_{\alpha,\beta}$ coefficient is the diagram obtained by erasing $V_n$ boxes from $p$ copies of Figure \ref{fig:f_n} (right panel), connecting the black decorations of the $i$-th $V_n$ box with the corresponding black decorations of the $\alpha(i)$-th $\overline V_n$ box, and connecting the white decorations of the $i$-th $V_n$ box with the corresponding white decorations of the $\beta(i)$-th $\overline V_n$ box. We shall denote the two $V_n$ boxes in the $i$-th copy of the diagram from the right panel of Figure \ref{fig:f_n} by $i^T,i^B$ (for, respectively, the box on the top row and the box on the bottom row). We also introduce the permutation
$$\mathcal S_{2p} \ni \delta := \prod_{i=1}^p (i^T, i^B),$$
permuting the top with the bottom row in the diagram. 

The resulting diagram $C_{\alpha,\beta}$ is, up to the pre-factor $d_n^{-p}$, a collection of loops, as follows:
\begin{itemize}
\item $\#\alpha$ loops of dimension $n$, corresponding to square-shaped decorations. The initial wiring is given by the identity permutation and the additional wiring is given by $\alpha$. 
\item $\#(\delta^{-1}\alpha)$ loops of dimension $k$, corresponding to diamond-shaped decorations. The initial wiring is given by the permutation 
$\delta$ and the additional wiring is given by $\alpha$.
\item $\#(\delta^{-1}\beta)$ loops of dimension $d_n$, corresponding to round-shaped decorations. The initial wiring is given by the permutation $\delta$ and the additional wiring is given by $\beta$.
\end{itemize}
Putting everything together, we get
$$\E f_n^p = d_n^{-p}\sum_{\alpha, \beta \in \mathcal S_{2p}} n^{\#\alpha}  k^{\#(\delta^{-1}\alpha)} d_n^{\#(\delta^{-1}\beta)} \Wg(nk,\alpha^{-1}\beta).$$
Using the asymptotic formula \eqref{eq:Wg-asymptotic} for the Weingarten function from Lemma \ref{lem:Wg-asymptotic} and the scaling $d_n \sim \lambda n$, we get 
$$\E f_n^p = (1+o(1))
\sum_{\alpha, \beta \in \mathcal S_{2p}} n^{p-(|\alpha| + |\alpha^{-1}\beta| + |\beta^{-1}\delta|)}  k^{-|\alpha^{-1}\beta| -|\delta^{-1}\alpha|} \lambda^{p-|\delta^{-1}\beta|} \Moeb(\alpha^{-1}\beta).$$
Now, we study the asymptotics of this expression as a function of $n$ when $n\to \infty$ (the other parameters being fixed, see \eqref{eq:asymptotic-regime}):
\begin{equation}
\text{exponent of } n  = p-(|\alpha| + |\alpha^{-1}\beta| + |\beta^{-1}\delta|) \leq p-|\delta| =0,
\end{equation}
where we have used the triangle inequality
$$|\alpha| + |\alpha^{-1}\beta| + |\beta^{-1}\delta| \geq |\delta| = p.$$
This inequality is saturated when $\id \to \alpha \to \beta \to \delta$ is a geodesic in $\mathcal S_{2p}$. Thus,
 \begin{equation}
\lim_{n\to\infty} \E f_n^p=\sum_{\mathrm{id}\to \alpha \to \beta  \to \delta}  k^{-|\delta^{-1}\alpha| - |\alpha^{-1}\beta|}  \lambda^{p-|\delta^{-1}\beta|} \Moeb(\alpha^{-1}\beta).
\end{equation}

Now, $\delta$ is a product of $p$ disjoint transpositions $\delta=\prod_{i=1}^p \tau_i$, where $\tau_i=(i^T,i^B)$. Permutations $\alpha$ and $\beta$ lie on the geodesic $\mathrm {id}\to \delta$ if and only if there exist two subsets $\emptyset \subseteq A \subseteq B\subseteq [p]$ such that $\alpha=\prod_{i\in A} \tau_i$ and $\beta=\prod_{i\in B} \tau_i$. Since $\alpha^{-1}\beta$ is a product of $|B\backslash A|$ transpositions of disjoint support, then $\Moeb(\alpha^{-1}\beta)=(-1)^{|B\backslash A|}$. Taking this into account we get 
\begin{equation}\begin{aligned}
\lim_{n\to\infty}\E f_n^p&= \sum_{\emptyset \subseteq A \subseteq B\subseteq [p]}  k^{-p+|A|-|B\backslash A|}  \lambda^{|B\backslash A|+|A|} (-1)^{|B\backslash A|}\\
&= k^{-p}\sum_{\emptyset \subseteq A \subseteq B\subseteq [p]}  (k\lambda)^{|A|}\left(\frac {-\lambda}k\right)^{|B\backslash A|}\\
\end{aligned}
\end{equation}
Using the multinomial identity 
$$ \sum_{\emptyset \subseteq A \subseteq B\subseteq [p]}  x^{|A|}y^{|B\backslash A|}=(1+x+y)^p,$$
we obtain the desired result
$$\lim_{n\to\infty}\E f_n^p
= \left(\lambda+\frac 1 k -\frac \lambda {k^2}\right)^p.$$
\end{proof}

\begin{theorem}\label{thm:bound-s1}
	Consider a sequence of random quantum channels $\Phi_n:M_{d_n}(\CC) \to M_n(\CC)$ given by 
	\[\Phi_n(X)= [\id_n \otimes \Tr_k ](V_nXV_n^*)=\sum_{i=1}^k A_i X A_i^*,\]
	where $k$ is a fixed constant, $d_n \sim \lambda n$ for another constant $\lambda \in (0,k)$, and $V_n :\CC^{d_n} \to \CC^n \otimes \CC^k$ is a Haar-distributed random isometry. Then, if $F_n$ is the super-operator associated with the channel $\Phi_n$, we have that,  $\forall \epsilon>0$,
$$\lim_{n \to \infty}\mathbb{P}\left[\|F_n\|_\infty \geq \sqrt{\lambda+\frac 1 k -\frac \lambda {k^2}}-\epsilon  \right]=1.$$
If, moreover, $d_n = (1+O(n^{-2}))\lambda n$, the above probability is lower bounded, at fixed $n$, by $1-\epsilon^{-2}O(n^{-2})$, and thus, almost surely
\begin{equation}\label{eq:almost-sure-LB}
\liminf_{n \to \infty} \|F_n\|_\infty \geq \sqrt{\lambda+\frac 1 k -\frac \lambda {k^2}}.
\end{equation}
\end{theorem}
\begin{proof}
Recall from Proposition \ref{prop:limit-overlap} that the quantity $f_n = \Omega_{d_n}^* F_n^* F_n \Omega_{d_n} = \|F_n \Omega_{d_n}\|^2$ converges in probability to $\lambda+\frac 1 k -\frac \lambda {k^2}$; since $\Omega_{d_n}$ is a unit vector, $\sqrt{f_n}$ is a lower bound for the operator norm of $F_n$, thus proving the main claim. The quantitative bound at fixed $n$ follows from Chebyshev's inequality and the observation that all the approximations in the proof of Proposition \ref{prop:limit-overlap} are of order $1+O(n^{-2})$. Indeed, note that if the geodesic inequalities are not saturated, there is a jump of one order in the exponent of $n$: this follows that the parity of the map $\sigma \mapsto |\alpha^{-1}\sigma| + |\sigma^{-1}\beta|$ is constant, for any permutations $\alpha, \beta$. Moreover, note that a similar jump in the asymptotic behavior of the Weingarten function follows from \eqref{eq:Wg-asymptotic}. Finally, the almost sure convergence follows from the Borel-Cantelli lemma \cite[Theorem 3.18]{kallenberg2002foundations}.  
\end{proof}

\begin{remark}
If, instead of analyzing $\langle \Omega_{d_n}, F_n^*F_n \Omega_{d_n} \rangle$ in the proof above, one looks at $\langle \Omega_n, F_nF_n^* \Omega_n \rangle$, using the isometry property $V_n^*V_n = I_d$ (which is equivalent to the trace-preservation property of the quantum channel $\Phi_n$), one would get the asymptotic lower bound $\|F\| \geq \lambda$. Note that this bound is worse that the one in the statement of the theorem, since $\lambda < k$.
\end{remark}

\begin{remark}
Instead of just computing the limiting overlap $\langle \Omega_{d_n}, F_n^*F_n \Omega_{d_n} \rangle$, it would be interesting to analyze the full limiting eigenvalue distribution of the random matrix $F_n^*F_n$; however, this seems to be out or reach at the current time with the moment techniques used in the proof of Proposition \ref{prop:limit-overlap} and with the Weingarten calculus as the main tool. 
\end{remark}

\section{Upper bound on the norm of the restriction}\label{sec:UB}

This section contains the second result needed for the spectral gap, an upper bound on the second singular value of the super-operator corresponding to a random quantum channel. The proof is largely inspired by G.~Pisier's work \cite{pisier2014quantum}, where he gave a different point of view on Hastings' quantum expander result from \cite{hastings2007random}. Since the model of random quantum channels we consider here is different (random isometries as opposed to random mixed unitary channels in \cite{hastings2007random,pisier2014quantum}), we going to give a self-contained presentation, and, on the way, slightly generalize some of the technical results from (the Appendix) of \cite{pisier2014quantum}. 

The main result of this section, Theorem \ref{thm:bound-s2}, which is the upper bound on the second singular value of the super-operators, is a consequence of two propositions: a comparison of our random model with a Gaussian model (Proposition \ref{prop:compare-with-Gaussian}) and a bound on the Gaussian model (Proposition \ref{prop:bound-Gaussian-model}). 

We start with the result relating the sequence of super-operators $F_n$ to a Gaussian model. We would like to point out that such techniques have been used to great success in Random Matrix Theory (e.g.~\cite{tao2011random}). The idea here is that, in Proposition \ref{prop:limit-overlap} we have identified the vector $\Omega_{d_n}$ as having a large overlap with a right-singular-vector corresponding to a (possibly) large singular value. To show a spectral gap, we will prove that, the restriction of $F_n$ the orthogonal complement of $\Omega_{d_n}$ has a relatively small norm. 

Let us recall the definition of the \emph{Ginibre ensemble} of random matrices. A $M \times N$ random matrix $X$ having i.i.d.~entries distributed as a centered complex Gaussian random variable with variance $\sigma^2$ is said to have a Ginibre distribution with parameters $(M,N;\sigma^2)$; we write $X \in \mathrm{Gin}(M,N;\sigma^2)$. We shall need the following lemma, which can be found in the Appendix of \cite{pisier2014quantum} in the case of square matrices. We generalize it here for rectangular matrices and we explicitly characterize the large dimension limit of the constants.

\begin{lemma}\label{lem:expected-tensor-positive-part}
For integers $M \geq N$, let $Y \in \mathrm{Gin}(M,N;1/M)$ be a random Ginibre matrix and denote $H:=\E(|Y|\otimes |\overline{Y}|) \in M_{N^2}(\CC)$.
Then,
\begin{equation}\label{eq:def-chi}
H= \omega_{N}+\chi_{M,N} (I_{N^2}-\omega_{N}) \qquad \text{ with } \qquad \chi_{M,N} = \frac{\E\|Y\|_1^2-1}{N^2-1}.
\end{equation}
For all $M,N$, we have $\chi_{M,N} \geq 1/(N+1)>0$. Moreover, in the limit where $N \to \infty$ and $M \sim cN$ for some constant $c\geq 1$, 
\begin{equation}\label{eq:def-chi-c}
\lim_{N \to \infty} \chi_{cN,N} = \chi_c:=c^{-1}\left[ \int_a^b \frac{\sqrt{(x-a)(b-x)}}{2\pi \sqrt x} \mathrm{d} x\right]^2,
\end{equation}
where $a=(\sqrt c - 1)^2$ and $b=(\sqrt c + 1)^2$. The function $c \mapsto \chi_c$ is increasing and we have  $\chi_1 = (8/(3\pi))^2$ and $\lim_{c \to \infty}\chi_c=1$.
\end{lemma}
\begin{proof}
The key to the proof is that the matrix $H$ is $U\otimes \overline U$-invariant, i.e.~for all unitary matrices $U \in \mathcal U_N$, 
$$(U \otimes \overline U) H (U \otimes \overline U)^* = \E(U|Y|U^*\otimes \overline U |\overline{Y}| \overline U^*)  = \E(|YU^*|\otimes |\overline{YU^*}|)= \E(|Y'|\otimes |\overline{Y'}|)  = H.$$
We can thus integrate the relation above over $U$, a procedure known as ``twirling'' in quantum information theory \cite{werner1989quantum}, which is yet another application of the Weingarten calculus: 
$$H = \int_{U \in \mathcal U_N} (U \otimes \overline U) H (U \otimes \overline U)^* \mathrm{d}U = \Tr(H \omega_N) \omega_N + \frac{\Tr[H (I_{N^2}-\omega_N)]}{N^2-1}(I_{N^2} - \omega_N).$$
In our case, we have 
\begin{align*}
\Tr(H \omega_N) &= N^{-1} \E \Tr |Y|^2 = N^{-1} \E \Tr Y^*Y = N^{-1} MN M^{-1} = 1 \\ 
\Tr H &= \E[ \Tr|Y| \Tr |\overline Y| ] = \E[(\Tr|Y|)^2] = \E \|Y\|_1^2,
\end{align*}
and the conclusion follows with the announced value of $\chi$. The inequality at fixed $M,N$ follows by writing
$$\chi_{M,N} = \frac{\E\|Y\|_1^2-1}{N^2-1} \geq \frac{\E\|Y\|_2^2-1}{N^2-1} = \frac{N-1}{N^2-1}.$$
The limiting case is a consequence of the formula for the  Marchenko-Pastur density, see \cite[Eq.~(3.1.1)]{bai2010spectral}:
$$\mathrm{d}\mathrm{MP}_c=\max (1-c,0)\delta_0+\frac{\sqrt{(b-x)(x-a)}}{2\pi x} \; \mathbf{1}_{[a,b]}(x) \, \mathrm{d}x,$$
with $a = (1-\sqrt c)^2$ and $b=(1+\sqrt c)^2$. Indeed, $M^{1/2}Y$ has standard Gaussian distribution, and thus
$$N^{-1}\left(M^{1/2}Y\right)^*\left(M^{1/2}Y\right) \to \mathrm{MP}_c,$$
where $\mathrm{MP}_c$ is the Marchenko-Pastur distribution of parameter $c$, and the convergence above holds in moments. Hence, 
$$\frac 1 N \sum_{i=1}^N s_i(Y) \to c^{-1/2} \int \sqrt x\  \mathrm{d}\mathrm{MP}_c(x),$$
from which the conclusion follows. 
\end{proof}

We now state our result, comparing the Schatten norms of our random quantum channel model with the corresponding norms of un-correlated Gaussian matrices. Note that the following result holds for fixed parameters $n,d,k$. 

\begin{proposition} \label{prop:compare-with-Gaussian}
    Consider a random quantum channel $\Phi:M_{d}(\CC) \to M_n(\CC)$ defined by a Haar-random isometry and let $F$ be the corresponding super-operator \eqref{eq:super-operator} associated with the channel $\Phi$.
Then, for any $p\geq 1$ and any $1\leq q \leq \infty$ we have
\begin{equation*}
\mathbb E \left\|F(I_{d^2}-\omega_{d})\right\|_q^p \leq  (2 / \chi_{nk,d})^p \,\mathbb E  \left\|\sum_{i= 1} ^k Y_i \otimes Z_i \right\|^p_q,
\end{equation*}
where $\{Y_i,Z_i\}_{i=1}^k$ are independent Ginibre random matrices of parameters $(n, d; (nk)^{-1})$, and $\chi$ are the constants from \eqref{eq:def-chi}.
\end{proposition}
\begin{proof}
We will adapt some of the ideas present in the Appendix of \cite{pisier2014quantum}. 

Let $Y \in \mathrm{Gin}(nk, d; (nk)^{-1})$ be a Ginibre random matrix; we recall that this means that the entries of $Y$ are i.i.d.~complex Gaussian variables with mean $0$ and variance $1/(nk)$. Using the unitary invariance of the Gaussian distribution, we can write its polar decomposition $Y = V |Y|$, where $V \in M_{nk, d}(\mathbb C)$ is a \emph{Haar-distributed} random isometry and $|Y|\in M_{d}(\mathbb C)$ is positive semidefinite, such that $V$ and $|Y|$ are independent random matrices (see \cite[Lemma 4.3.10]{hiai2000semicircle} for a proof of the square case).

Let $\mathcal E$ be the conditional expectation operator with respect to the $\sigma$-algebra generated by isometry part $V$. We have 
\begin{equation*}
\mathcal E(Y\otimes \overline{Y}) = \mathcal E(V|Y|\otimes \overline{V|Y|})=(V \otimes \overline{V})\mathbb E(|Y|\otimes |\overline{Y}|).
\end{equation*}

Using Lemma \ref{lem:expected-tensor-positive-part}, we obtain
$$\mathcal E(Y\otimes \overline{Y}) = (V \otimes \overline V) \left[\omega_d + \chi_{nk,d}(I_{d^2}-\omega_d)\right].$$
Taking the $(i,i)$ block elements of each side and multiplying on the right with $I_{d^2}-\omega_d$ yields
$$\chi_{nk,d}^{-1}\mathcal E\left[ \sum_{i=1}^k Y_i \otimes \overline Y_i \right] (I_{d^2}-\omega_d) = F(I_{d^2}-\omega_d),$$
where $Y_1, \ldots, Y_k \in M_{n,d}(\CC)$ are the blocks of $Y$:
$$Y = \sum_{i=1}^k e_i \otimes Y_i.$$
Since $\|\cdot \|_q^p$ is a convex function, after applying Jensen's inequality, we get

\begin{equation*}
\mathbb E \left \| F(I_{d^2} - \omega_d) \right\|_q^p\leq \chi_{nk,d}^{-p} \mathbb E  \left\|  \sum_{i= 1} ^k (Y_i \otimes  \overline{Y_i})(I_{d^2} - \omega_d)  \right\|^p_q
\end{equation*}

We shall now use a decoupling argument from the Appendix of \cite{pisier2014quantum}, to go from the left-hand-side of the equation above to the expression in the statement. We state it here without proof, since there is no difference between the case of square matrices discussed in \cite{pisier2014quantum} and our (more general) situation involving rectangular blocks $Y_i$. we just note that one of the key ingredients of the proof is that, for all $i$, $\E(Y_i \otimes  \overline{Y_i})(I_{d^2} - \omega_d) = 0$.
\begin{lemma} Let $Y_1, \ldots, Y_k$ be independent $n\times d$ Ginibre matrices, and consider independent copies $Z_1, \ldots, Z_k$ having the same distributions. Then,
\begin{equation*}
\mathbb E  \left\|  \sum_{i= 1} ^k (Y_i \otimes  \overline{Y_i})(I_{d^2} - \omega_d)\right\|^p_q 
\leq 2^p \mathbb E  \left\|  \sum_{i= 1} ^k  (Y_i \otimes Z_i)(I_{d^2} - \omega_d)\right\|^p_q 
\end{equation*}
\end{lemma}

Using the lemma, we conclude
$$\mathbb E \left\|F(I_{d^2}-\omega_d)\right\|_q^p \leq  \left(\frac{2}{\chi_{nk,d}}\right)^p \,\mathbb E  \left\|\sum_{i= 1} ^k Y_i \otimes Z_i \right\|^p_q.$$
\end{proof}

\bigskip

We now move on to the second technical result of this section, a bound on the moments of the Gaussian model. 

\begin{proposition}\label{prop:bound-Gaussian-model}
Let $Y_1, \ldots, Y_k, Z_1, \ldots, Z_k$ be independent Ginibre random matrices of parameters $(n, d; (nk)^{-1})$, where $n,d,k$ are fixed integers. Then, for all even integers $p \geq 2$,

\begin{equation*}
\mathbb E  \left\| \sum_{i= 1} ^{k} Y_i \otimes  Z_i\right\|_\infty^p  \leq  n^{2} \left( \frac{(1+\sqrt{\lambda})^2}{\sqrt k} +\epsilon + \beta\sqrt{\frac p n} \right)^p
 \end{equation*}
 where $\epsilon$ and $\beta$ are functions of $n,d,k$ with the property that, in the asymptotic regime \eqref{eq:asymptotic-regime}, $\epsilon(n) \to 0$ and $\beta(n)$ is bounded, as $n \to \infty$.

\end{proposition}
\begin{proof}
	We follow the proof of \cite[Theorem 16.6]{pisier2012grothendieck}. Let $Y \in \mathrm{Gin}(n,d;(nk)^{-1})$ be another Ginibre matrix as in the statement. By \cite[Theorem 5.11]{bai2010spectral}) we know that, almost surely, in the asymptotic regime \eqref{eq:asymptotic-regime}, 
	$$\lim_{n\to\infty} \|Y\|_\infty=\frac{1+\sqrt{\lambda}}{\sqrt k}.$$
	Let us then define the function $\epsilon = \epsilon(n,d,k)$ by
$$\E\|Y\|_\infty = \frac{1+\sqrt{\lambda}}{\sqrt k}+\epsilon,$$
	so that $\epsilon \to 0$ when $n \to \infty$ in the regime \eqref{eq:asymptotic-regime}. Again by concentration of measure arguments (see \cite[Chapter 2]{pisier1986probabilistic}) the function $\beta = \beta(n,d,k,p)$ such that for any $n\ge 1$ and $p\ge 2$ we have 
	\begin{equation}\label{eq:moment-norm-Y-epsilon-beta}
	(\E\|Y\|_\infty^p)^{1/p} =  \E\|Y\|_\infty + \beta \sqrt{\frac p n} \le \frac{1+\sqrt{\lambda}}{\sqrt k} +\epsilon + \beta\sqrt{\frac p n}
	\end{equation}
	is such that, in the regime \eqref{eq:asymptotic-regime}, $\beta$ is bounded uniformly in $p$, as $n \to \infty$. 

	Let us write 
	$$X:= \sum_{i=1}^k Y_i \otimes Z_i \in M_{n^2, d^2}(\CC).$$
	We claim that one has, for any even integer $p\ge 2$,
	\begin{equation}\label{eq:claim-X-k-Y}
	\E \Tr |X|^p  \le k^{p/2} (\E \Tr |Y|^p)^2.
	\end{equation}
	Let us defer the proof of the claim for later, and now use it together with \eqref{eq:moment-norm-Y-epsilon-beta} to get:
	\begin{align*}
	\E\|X\|_\infty^p &\leq \E\|X\|_p^p \leq k^{p/2} (\E\|Y\|_p^p)^2 \leq k^{p/2} (\min(n,d) \E\|Y\|_\infty^p)^2 \\
	&\leq n^2k^{p/2}\left( \frac{1+\sqrt{\lambda}}{\sqrt k} +\epsilon + \beta\sqrt{\frac p n} \right)^{2p}\\
	&= n^2 \left( \frac{(1+\sqrt{\lambda})^2}{\sqrt k} +\epsilon' + \beta'\sqrt{\frac p n} \right)^p,
	\end{align*}
	which is the desired statement, for some modified functions $\epsilon',\beta'$ enjoying the same asymptotic properties.

We verify now the claim \eqref{eq:claim-X-k-Y}. Let $p=2m$. We develop
$$	\E\Tr|X|^p =\E\Tr(X^*X)^m = \sum_{i_1, \ldots, i_m,j_1,\ldots, j_m =1}^k ( \E\Tr Y^*_{i_1}Y_{j_1}\cdots Y^*_{i_m}Y_{j_m})^2.$$
Note that, using the Wick formula to evaluate the Gaussian expectation,  the only non-vanishing terms in this sum correspond to certain pairings that guarantee that $\E\Tr Y^*_{i_1}Y_{j_1}\cdots Y^*_{i_m}Y_{j_m} \geq 0$:
$$\E\Tr Y^*_{i_1}Y_{j_1}\cdots Y^*_{i_m}Y_{j_m} = (nk)^{-m}\sum_{\alpha \in \mathcal S_p \, : \, j = i \circ \alpha} n^{\# \alpha} d^{\#(\alpha^{-1}\gamma)},$$
where $\gamma = (1\, 2\, \cdots\, m) \in \mathcal S_m$ is the full-cycle permutation. Moreover, by H\"older's inequality for the trace, we have, for all $m$-tuples $i$ and $j$,
$$|\E\Tr Y^*_{i_1}Y_{j_1}\cdots Y^*_{i_m}Y_{j_m}| \leq \E \|Y^*_{i_1}Y_{j_1}\cdots Y^*_{i_m}Y_{j_m}\|_1 \leq \E \|Y\|_{2m}^{2m}= \E\Tr|Y|^p.$$
From these observations, we find
\begin{equation}\label{eq:final-ineq-X-Y}
\E \Tr |X|^p \leq \E \Tr |Y|^p \sum_{i,j} \E\Tr Y^*_{i_1}Y_{j_1}\cdots Y^*_{i_m} Y_{j_m},
\end{equation}
where we have used the crucial positivity property mentioned before. Evaluating the last sum gives
\begin{align*}
\sum_{i,j} \E\Tr Y^*_{i_1}Y_{j_1}\cdots Y^*_{i_m} Y_{j_m}&= k^m \E\Tr \left(\sum_{i_1=1}^k k^{-1/2}Y_{i_1}\right)^* \cdots \left(\sum_{j_m=1}^k k^{-1/2}Y_{j_m}\right) \\
&= k^m\E\Tr (\hat Y^*\hat Y)^m= k^{p/2}\E\Tr|\hat Y|^p,
\end{align*}
where $\hat Y := k^{-1/2} \sum_{i=1}^k Y_i$ has the same distribution as $Y$; hence \eqref{eq:final-ineq-X-Y} implies \eqref{eq:claim-X-k-Y}, and the proof is complete.
    \end{proof}

We can now state the main result of this section, an almost-sure upper bound on the norm of the super-operators $F_n$, restricted on the space orthogonal to the maximally entangled state $\Omega_{d_n}$. 

\begin{theorem} \label{thm:bound-s2}

Consider a sequence of random quantum channels $\Phi_n:M_{d_n}(\CC) \to M_n(\CC)$ defined by  Haar-random isometries and let $F_n$ be the corresponding super-operators \eqref{eq:super-operator} associated with the channels $\Phi_n$, in the asymptotic regime \eqref{eq:asymptotic-regime}.
Then, for any $p\geq 1$ 
\begin{equation*}
\left(\mathbb E \left\|F(I_{d_n^2}-\omega_{d_n})\right\|_\infty^p\right)^{1/p} \leq n^{2/p} \left( g_{k,\lambda} + \epsilon + \beta \sqrt{\frac p n}\right),
\end{equation*}
where $\epsilon = \epsilon(n) \to 0$ and $\beta = \beta(n)$ is bounded as $n \to \infty$, and (see \eqref{eq:def-chi-c} for the definition of $\chi_\cdot$)
\begin{equation}\label{eq:def-g-k-lambda}
    g_{k,\lambda}:=  \frac{2(1+\sqrt \lambda)^2}{\chi_{k/\lambda}\sqrt k}.
\end{equation}
In particular, we have that, almost surely, 
\begin{equation}\label{eq:almost-sure-UB}
\limsup_{n \to \infty} \left\|F_n(I_{d_n^2}-\omega_{d_n})\right\|_\infty \leq g_{k,\lambda} .
\end{equation}
\end{theorem}
\begin{proof}
The first claim is obtained by combining Proposition \ref{prop:compare-with-Gaussian} for $q=\infty$ and Proposition \ref{prop:bound-Gaussian-model}, where the error $\chi_{nk,d_n}-\chi_\lambda$ is absorbed in $\epsilon$. To obtain the almost-sure bound on the norm, apply Borel-Cantelli after noticing that, for $\delta > 0$, the probability
$$\mathbb{P}\left[ \left\|F_n(I_{d_n^2}-\omega_{d_n})\right\|_\infty \geq (1+\delta) g_{k,\lambda} \right] \leq n^2 (1+\delta)^{-p}\left(1 + \epsilon + \beta \sqrt{\frac p n} \right)^p $$
is an $O(n^{-2})$ for, say, $p=(5/\delta) \log n$.
\end{proof}
\begin{remark}\label{rem:g-k-lambda-asymptotics-k-large}
Assuming that $\lambda$ is being kept fixed, let us analyze the asymptotic behaviour of $g_{k,\lambda}$ as $k \to \infty$; note that this situation corresponds to a global asymptotic regime where $1 \ll k \ll n$ (compare with \eqref{eq:asymptotic-regime}). We start from ``law of large numbers'' for the Marchenko-Pastur distribution which states that, as the parameter $c$ grows, 
$$D_{c^{-1}}\mathrm{MP_c} \to \delta_1.$$
Hence, when $k/\lambda \to \infty$, we have $\chi_{k/\lambda} \sim 1$. Thus, the behaviour of $g_{k,\lambda}$ as $\lambda$ is fixed and $k \to \infty$ is given by
$$g_{k,\lambda} \sim \frac{2(1+\sqrt \lambda)^2}{\sqrt{k}}.$$
\end{remark}

\section{Spectral gap of a random quantum channel}\label{sec:spectral-gap}

We state now the first main result of this work, an asymptotic lower bound on the singular value gap for the sequence of super-operators $F_n$.

\begin{theorem} \label{thm:gap-sv}
Consider a sequence of random quantum channels $\Phi_n:M_{d_n}(\CC) \to M_n(\CC)$ defined by  Haar-random isometries and let $F_n$ be the corresponding super-operators \eqref{eq:super-operator} associated with the channels $\Phi_n$, in the asymptotic regime \eqref{eq:asymptotic-regime}. Then, almost surely as $n \to \infty$, we have the following gap between the largest two singular values of $F_n$:
\begin{equation}\label{eq:gap-sv}
\liminf_{n \to \infty} \left[ s_1(F_n) - s_2(F_n)\right] \geq \sqrt{\lambda + \frac 1 k - \frac{\lambda}{k^2}} -  g_{k,\lambda},
\end{equation}
where the constant $g_{k,\lambda}$ was defined in \eqref{eq:def-g-k-lambda}.
\end{theorem}
\begin{proof}
The result follows from Theorems \ref{thm:bound-s1} and \ref{thm:bound-s2} (more precisely from equations \eqref{eq:almost-sure-LB} and \eqref{eq:almost-sure-UB}) after using the perturbation result for singular values from \cite[Problem III.6.4]{bhatia1997matrix}:
$$s_1(F) - s_2(F) \geq s_1(F) - s_1(F(I-\omega)) = \|F\|_\infty - \|F(I-\omega)\|_\infty.$$
\end{proof}
\begin{remark}
The lower bound from \eqref{eq:gap-sv} is not really explicit, since the quantity $\chi_{k/\lambda}$ appearing in the definition of $g_{k,\lambda}$ is defined, for general $k,\lambda$ in terms of the elliptic integral \eqref{eq:def-chi-c}. However, we can lower bound it by an explicit formula, as follows:
\begin{align*}
\liminf_{n \to \infty} \left[ s_1(F_n) - s_2(F_n)\right] &\geq \sqrt{\lambda + \frac 1 k -\frac{\lambda}{k^2}} -  \frac{2(1+\sqrt \lambda)^2}{\chi_{k/\lambda}\sqrt k}\\
&\geq \sqrt{\lambda + \frac 1 k -\frac{\lambda}{k^2}} -  \frac{2(1+\sqrt \lambda)^2}{\chi_1 \sqrt k} \\
&=  \sqrt{\lambda + \frac 1 k -\frac{\lambda}{k^2}} -  \frac{9\pi^2(1+\sqrt \lambda)^2}{32 \sqrt k}.
\end{align*}
\end{remark}
\begin{remark}
One can also obtain a rather simple upper bound on the asymptotic singular value gap as follows: 
$$s_1(F) - s_2(F) \leq \|F\|_\infty \leq \|F\omega\|_\infty + \|F(I-\omega)\|_\infty.$$
But recall that in the case of the super-operator sequence $F_n$, $\|F_n\omega_n\|_\infty$ is precisely the square root of the quantity $f_n$ studied in Proposition \ref{prop:limit-overlap}. It follows that
$$\limsup_{n \to \infty} \left[ s_1(F_n) - s_2(F_n)\right] \leq \sqrt{\lambda + \frac 1 k - \frac{\lambda}{k^2}} +  g_{k,\lambda}.$$
\end{remark}

We plot in Figure \ref{fig:gap-sv-lambda-1} the bound obtained above in the case where $\lambda=1$. Notice that the bound approaches $1$ as $k \to \infty$ and that, for small values of $k$, the bound is trivial (negative); our bound becomes non-trivial for $k\approx 63.52$. 

\begin{figure}[ht]
\centering
\includegraphics[scale=.6]{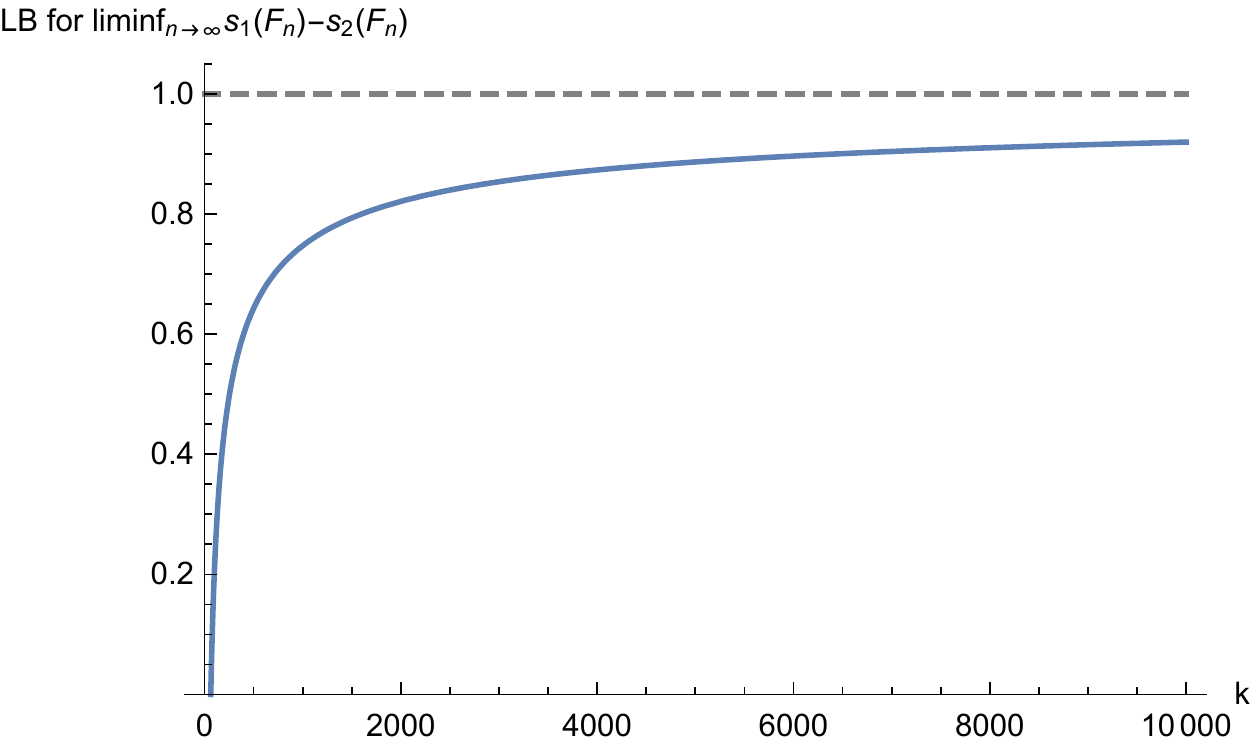} \qquad \includegraphics[scale=.6]{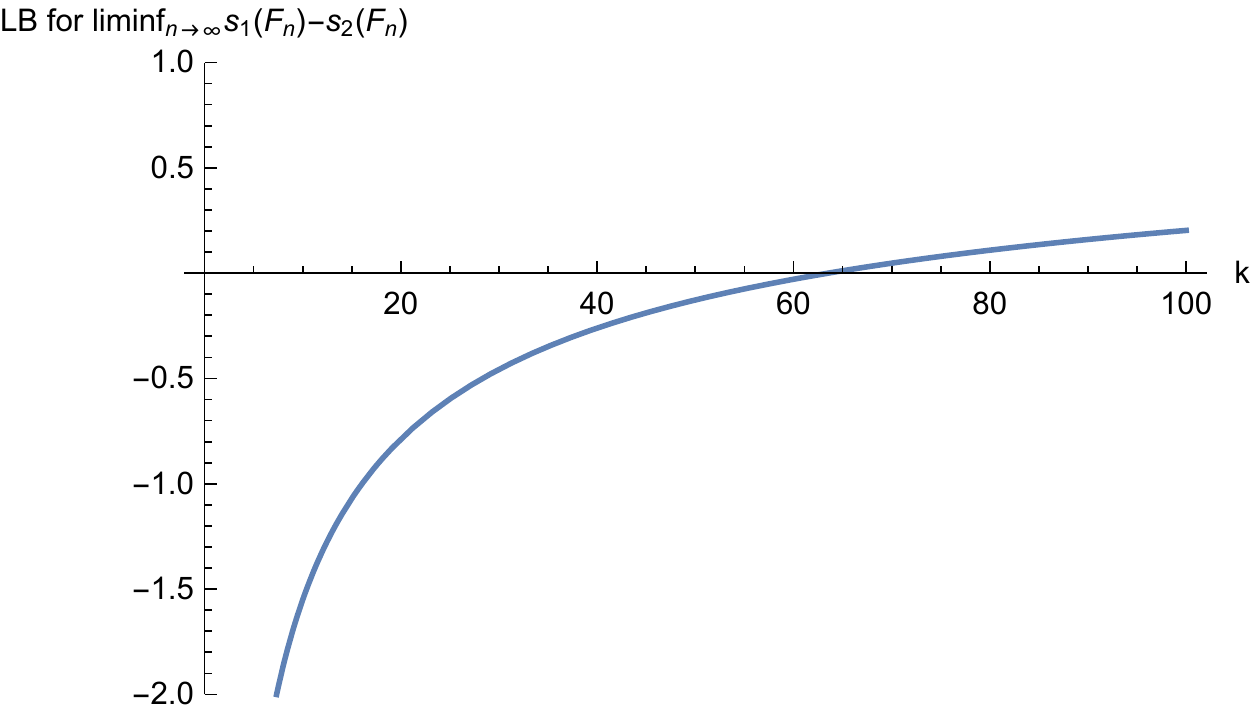}
\caption{Plot of the asymptotic lower bound on the singular value gap as a function of $k$, in the case where $\lambda = 1$. Left panel: large values of $k$; right panel: the regime where the lower bound becomes non trivial, i.e.~$k \gtrsim 63.52$.}
\label{fig:gap-sv-lambda-1}
\end{figure}

In the case where $d_n = n$,which corresponds to $\lambda =1$ in \eqref{eq:asymptotic-regime}, the super-operators $F_n$ are square matrices (of size $n^2$), and we can analyze their spectral gap. Note that the largest eigenvalue, called the Perron-Frobenius eigenvalue, is equal to 1 in this case \cite{evans1978spectral}, lower bounds on the spectral gap correspond to upper bounds on the modulus of the second eigenvalue. 

\begin{theorem} \label{thm:gap-ev}
Consider a sequence of random quantum channels $\Phi_n:M_{n}(\CC) \to M_n(\CC)$ defined by  Haar-random isometries and let $F_n$ be the corresponding super-operators \eqref{eq:super-operator} associated with the channels $\Phi_n$, in the asymptotic regime \eqref{eq:asymptotic-regime}. Then, almost surely as $n \to \infty$, the second largest (in absolute value) eigenvalue of $F_n$ is asymptotically upper bounded:
\begin{equation}\label{eq:gap-ev}
\limsup_{n \to \infty} |\lambda_2(F_n)| \leq \left(\sqrt{1 +  \frac{k-1}{k^2}} + g_{k,1}\right)g_{k,1}.
\end{equation}
\end{theorem}
\begin{proof}
Using Weyl's Majorant Theorem \cite[Theorem II.3.6]{bhatia1997matrix}, we have, for all $p>0$,
$$1 + |\lambda_2(F)|^p \leq s_1(F)^p + s_2(F)^p.$$
In our case, we use the asymptotic almost sure upper bounds ($\lambda=1$ below)
\begin{align*}
    \limsup_{n \to \infty} |s_1(F_n)| &\leq  \sqrt{\lambda + \frac 1 k - \frac{\lambda}{k^2}} +  g_{k,\lambda} > 1\\
    \limsup_{n \to \infty} |s_2(F_n)| &\leq  g_{k,\lambda} < 1.
\end{align*}
To get the best bound, we need to choose $p \to 0$. With this choice, the conclusion follows after using 
$$\lim_{p \to 0} \left( a^p + b^p - 1\right)^{1/p} = ab.$$
\end{proof}
\begin{remark}
Using the asymtptotic behaviour of $g_{k,\lambda}$ from Remark \ref{rem:g-k-lambda-asymptotics-k-large}, we have, in our present case, the following asymptotic lower bound for the spectral gap:
$$\liminf_{n \to \infty} 1-|\lambda_2(F_n)| \gtrsim 1-\frac{8}{\sqrt k}.$$
\end{remark}

We plot in Figure \ref{fig:UB-ev2} the upper bound from the theorem above, as a function of $k$. We observe numerically that the bound becomes smaller than 1 for $k \approx 168.5$.

\begin{figure}[ht]
\centering
\includegraphics[scale=.6]{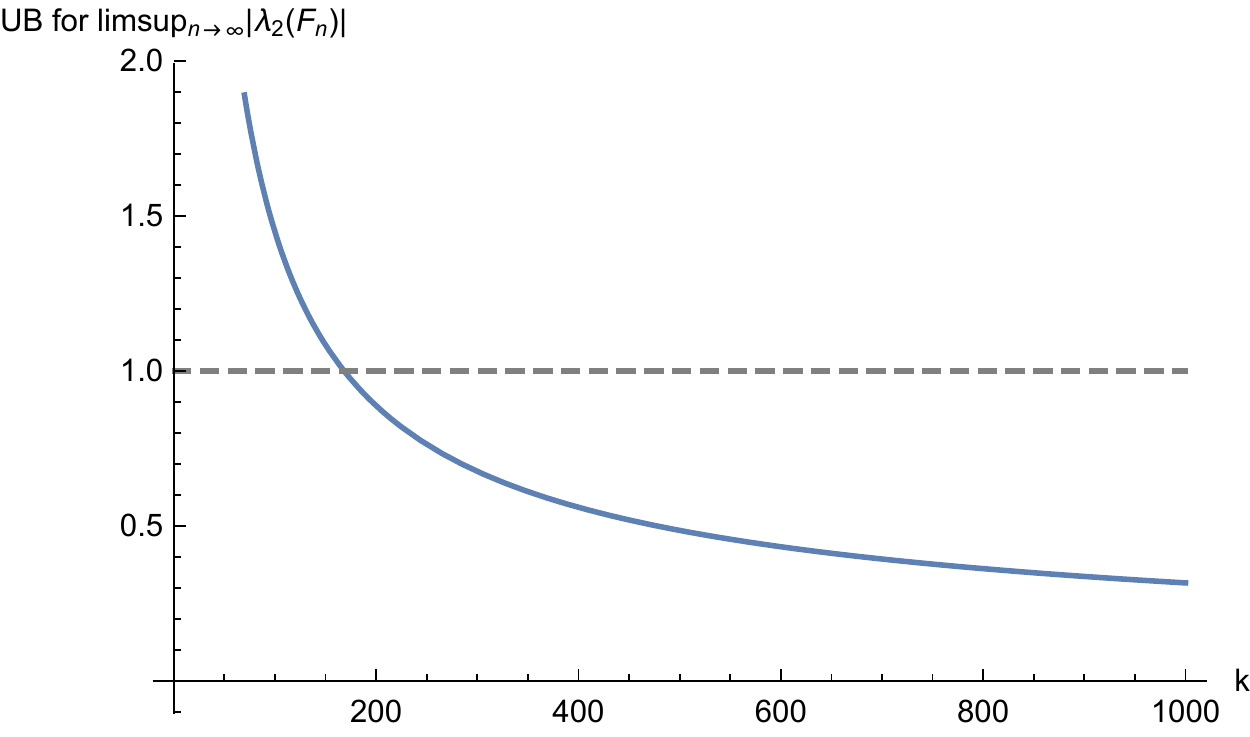}
\caption{Plot of the asymptotic upper bound on the absolute value of the second eigenvalue of the super-operators $F_n$.}
\label{fig:UB-ev2}
\end{figure}

\section{Random channels are quantum expanders}\label{sec:quantum-expander}

Inspired from classical combinatorial theory and computer science \cite{hoory2006expander}, quantum expanders have been introduced independently in \cite{ben2008quantum} and \cite{hastings2007entropy}. We follow here the more general definition suggested in \cite{hastings2007random}, which allows for non-unital quantum channels. 

\begin{definition}\label{def:quantum-expander}
A sequence of quantum channels $\Phi_n : M_n(\CC) \to M_n(\CC)$  is a \emph{quantum expander} if satisfies the following properties: 
\begin{itemize}
    \item $\Phi_n$ have Kraus rank at most $k$
    \item The sequence of second largest (in modulus) eigenvalues of $\Phi_n$ is upper bounded
    \item The (unique) invariant states of $\Phi_n$ have ``large entropy''.
\end{itemize}
\end{definition}

In order to show that a sequence of 
random quantum channels as in Section \ref{sec:random-quantum-channels} is a quantum expander, we just need to give an estimate on the entropy of the fixed point of the random channels (the upper bound on the second largest eigenvalue having been proved in Theorem \ref{thm:gap-ev}). Note that for unital quantum channels, this is not an issue, since the fixed point is the maximally mixed state, which has maximal entropy. In our case, we shall approximate the fixed point by the maximally mixed case, and bound the distance between the two by iterating the quantum channel.

\begin{lemma}\label{lem:approximate-fixed-point-by-identity}
Let $\Phi:M_n(\mathbb C)\to M_n(\mathbb C)$ be a quantum channel having a unique fixed point $\Lambda_\Phi$. Then,  for all $t \geq 1$,
\[\|\Phi^t(\I/n)-\Lambda_\Phi\|_2\leq 2 |\lambda_2(\Phi)|^t,\]
where $\|\cdot\|_2$ denotes the Hilbert-Schmidt norm.
\end{lemma}
\begin{proof}
The claim follows from the following:
 $$   \|\Phi^t(\I/n)-\Lambda_\Phi\|_2 =  \|\Phi^t(\I/n)-\Phi^t(\Lambda_\Phi)\|_2 \leq |\lambda_2(\Phi)|^t\|\I/n-\Lambda_\Phi\|_2 \leq 2 |\lambda_2(\Phi)|^t.$$
\end{proof}
We can now prove the main result of this section, that sequences of random quantum channels as in \eqref{eq:asymptotic-regime} are \emph{quantum expanders} in the sense of Definition \ref{def:quantum-expander}. Note that in the following theorem, we can assume that the fixed point of a sequence of random quantum channels is unique, see \cite[Theorem 4.4]{nechita2012random}.

\begin{theorem}\label{thm:entropy}
Let $k\geq 169$ be a fixed integer, and consider a sequence $\Phi_n:M_n(\mathbb C)\to M_n(\mathbb C)$ of random quantum channels as in \eqref{eq:asymptotic-regime}. Then, if  $\Lambda_n$ is the unique fixed point of $\Phi_n$,  
$$S(\Lambda_n) \geq \log k - o(1),$$
except with probability exponentially small in $n$.
\end{theorem}
\begin{proof}
Applying Jensen's inequality we get that for any quantum state  $\rho\in M_n(\CC)$
\[S(\rho)=-\tr(\rho\log(\rho))\geq -\log\tr(\rho^2).\]
We have
\begin{align}
    \nonumber \left|\tr[(\Phi_n^t(I/n))^2]-\tr[(\Lambda_n)^2)]\right|&\leq \|(\Phi_n^t(I/n))^2-\Lambda_n^2\|_1\\
    \label{eq:approxfixpoint}&\leq \|\Phi_n^t(I/n)-\Lambda_n\|_1 \|\Phi_n^t(I/n)-\Lambda_n\|_\infty\\
    \nonumber &\leq 2\|\Phi_n^t(I/n))-(\Lambda_n)\|_1\leq 4 |\lambda_2(\Phi_n)|^t.
\end{align}

Let $V_n$ be the isometry inducing $\Phi_n$ and define $g(V_n)=\tr((\Phi_n^t(I/n))^2)$. Then, $\mathbb{E}(g(V_n))$ can be seen as a particular case of the model introduced in \cite{Collins2013} that is the computation we are interested would be $\mathbb{E}(\tr(\rho_1^2))$ there, where the spin chain would have a left boundary condition $L=I$, the bulk is formed of $t$ systems and the left one of them is the accessible one, and right boundary condition is $R=I/n$. Following the same calculations there, we arrive to  
\begin{equation*}
    \mathbb{E}(g(V_n))=1/k+ O(n^{-1/5}),
\end{equation*}
and the Lipschitz constant of $g(V_n)$ is upper bounded by $4t$. Hence, applying a concentration result \cite{ledoux2001concentration}, we get that there exist constants $c'_1$, $c'_2$ such that \[g(V_n)=1/k+ O(n^{-1/5})\]
with probability larger than $1-c'_1e^{-c'_2n^{3/5}/t^2}\geq 1-c'_1e^{-c'_2n^{1/5}}$.\\
Putting this together with eq.~\eqref{eq:approxfixpoint} we get 
\[S(\Lambda_n)\geq -\log\tr(\Lambda_n^2)\geq -\log(4|\lambda_2(\Phi_n)|^t+1/k+ O(n^{-1/5}))\]
with probability $\geq 1-c_1e^{-c_2n^{1/5}}$, where we are taking into account the probability of $\Lambda_n$ being the only fixed point. Choosing $t$ growing with $n$, but not faster than $n^{1/5}$, the inequality from the statement follows. Since the sequence of random quantum channels $\Phi_n$ has fixed Kraus rank $k$ and we have already shown that the second largest eigenvalue (in absolute value) of the super-operators $F_n$ are bounded (see Theorem \ref{thm:gap-ev}), the conclusion follows.
\end{proof}

\section{Principle of maximum entropy for translationally-invariant matrix product states}\label{sec:TIMPS}

In this section we derive a principle of maximum entropy for infinite spin chains, for which the ground state is well-approximated \cite{hastings2007area} by \emph{matrix product states} (MPS) \cite{perez-garcia2006matrix}.  Ensembles of MPS were already considered in \cite{Garnerone2010a}, showing concentration of local observables, and in \cite{Collins2013}, showing concentration of the reduced density matrix with boundary conditions. Here, we focus in the case of infinite 1D translationally-invariant (TI) systems.

A translationally-invariant matrix product state (TI-MPS) is a state where the coefficients are given by a product of matrices in the following way
\[\psi=\sum_{i_1,\dots, i_N} \tr \left[ A_{i_N} 
\cdots A_{i_1} \right] 
e_{i_1}\otimes \cdots \otimes e_{i_N}\]
where $A_i\in \mathcal M_D(\CC)$ are some matrices. There is a close relation between a TI-MPS and the completely positive map $$\Phi(X)=\sum \limits_{i} A_i X A_i^*.$$ acting on $\mathcal M_D$. In fact, if we define the isometries $V = \sum_{i=1}^k A_i \otimes e_i$, there is a one to one correspondence between the set of TI-MPS with bond dimension $D$ and the set of quantum channels acting on $\mathcal M_D$. Moreover, for an infinite TI-MPS given by an isometry $V$ that induces a channel $\Phi$ with a unique fixed point $\Lambda_\Phi$, the reduced density matrix on $l$ sites is given by \[\rho_l=\mathcal E^l(\Lambda_\Phi)\] where   $\mathcal E(X)= V X V^*$ (note the absence of the partial trace operation in the definition of $\mathcal E$). Thus, there is a natural way to take an infinite random TI-MPS, with physical dimension $k$ and bond dimension $D$, by taking a random isometry (acting on $\mathcal M_D$ with $k$ Kraus operators. 

Set $D=d_n=n$ which corresponds to $\lambda=1$. In this case, the set of TI-MPS with bond dimension $D$ and physical dimension k\footnote{Note the different nomenclature with respect to the usual notation in tensor networks, $D$ is usually the bond dimension and $d$ is the physical dimension} has a one to one correspondence with the set of isometries introduced with fixed $k$. In this case the local tensors of the MPS are exactly the Kraus operators of the channel $\Phi$. 

For a fixed physical dimension $k$ and a fixed bond dimension $D$ we consider the ensemble of infinite TI-MPS induced by the random isometries $V:\CC^D\rightarrow\CC^D\times \CC^k$. In this case the ensemble of reduced density matrices over $l$ systems is given by 
\begin{equation}\label{eq:def-rho-l}
    \rho_l=\mathcal E^l(\Lambda_\Phi)
\end{equation}
where $\mathcal E(X)= V X V^*$, and $\Lambda_\Phi$ is the fixed point of $\Phi(X)=[\id_D \otimes \Tr_k](VXV^*)=\sum_{i=1}^k A_i X A_i^*$.

\begin{theorem} \label{tracerho2}
Let $k\geq 169$ and let $\rho_l$ be taken at random from the ensemble introduced above, such that $D\geq (t+l)^5$ and $t=t_D$ grows slower than $D^{1/5}$. Then
$\tr(\rho_l^2) \leq 1/k^l+ O(D^{-1/5}) $ except with probability exponentially small in $D$.
\end{theorem}

\begin{proof}
The idea of the proof is to approximate $\rho_l=\mathcal E^l(\Lambda_\Phi) $ by $\tilde\rho_l=\mathcal E^l(\Phi^t(\I/D))$ and that way we avoid the non explicit algebraic dependency between the channel $\mathcal E$ and the fixed point $\Lambda_\Phi$ of channel $\Phi$.
Indeed, using Lemma \ref{lem:approximate-fixed-point-by-identity} we have \[|\tr (\rho^2_l)- \tr(\tilde\rho_l^2)|\leq 2l \|\Lambda_\Phi- \Phi^t(\I/D)\|_1\leq 2l \sqrt D \|\Lambda_\Phi- \Phi^t(\I/D)\|_2\leq 2 l \lambda_2^t(\Phi)\leq 1/(k^l \sqrt D)\] In order to show the result we will follow \cite{Collins2013} to give a bound for $g(V)=\tr(\tilde \rho_l^2(V))$. Following the calculations done there we have that \begin{equation}\label{eq:rhol}
    \mathbb E[g(V)]=1/k^l+ O(D^{-1/5}),
\end{equation}
and the Lipschitz constant of $g(V)$ is upper bounded by $4(t+l)$.\\
Hence, applying a concentration result, we get that there exist constants $c_1$, $c_2$ such that \[\tr\left(\tilde \rho_l^2(V)\right)=1/k^l+ O(D^{-1/5})\]
with probability $1-c_1e^{-c_2D^{3/5}/(t+l)^2}\geq 1-c_1e^{-c_2D^{1/5}}$.\\
Putting this together with eq.~\eqref{eq:rhol} finishes the proof.

\end{proof}
\begin{figure}[ht]
\centering
\includegraphics[scale=1]{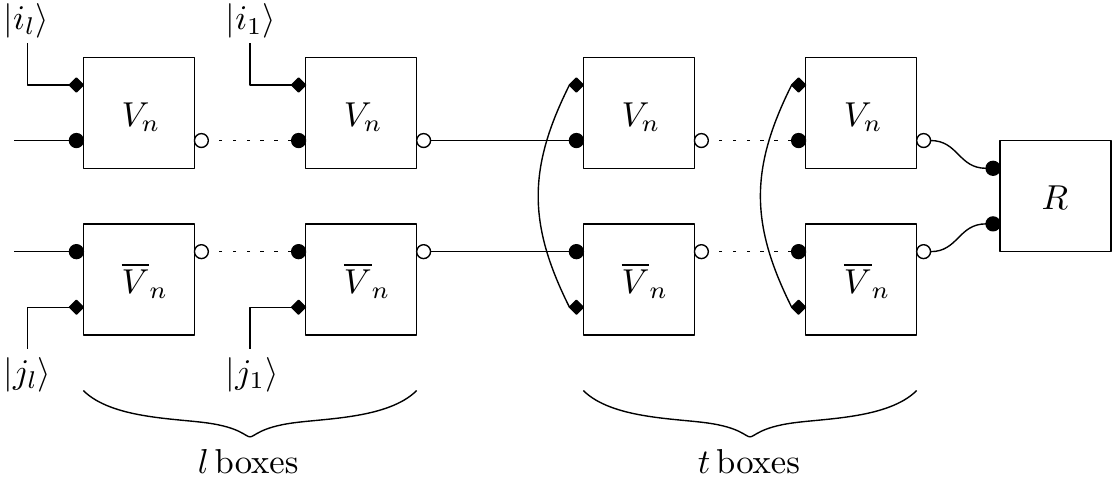}
\caption{A graphical representation of the approximation of a MPS with $R=I/D$ and open boundary conditions on the left.}
\label{fig:rhoapprox}
\end{figure}

In the same way as in \cite{Collins2013} we get two immediate consequences:

\begin{corollary}
Let $k\geq 169$ and let $\rho_l$ be taken at random from the ensemble introduced in \eqref{eq:def-rho-l}, where $l \ll D^{1/5}$. Then, with overwhelming probability as $D \to \infty$,
\[\left\|\rho_l-\frac {\I } {k^l} \right\|_\infty\leq O(D^{-1/10}).\]
\end{corollary}

\begin{corollary}
Let $k\geq 169$ and let $\rho_l$ be taken at random from the ensemble introduced in \eqref{eq:def-rho-l}, where $l \ll D^{1/5}$. Then the von-Neumann entropy verifies 
\[S(\rho_l)=l\log k- k^l O(D^{-1/5}),\] 
except with exponentially small probability in $D$.
\end{corollary}

\bigskip

\noindent {\it Acknowledgments.} C.E.G.G. is supported by Spanish MINECO (projects MTM2014-54240-P and MTM2017-88385-P) and MECD ``Jos\'e Castillejo'' program (CAS16/00339). I.N.'s research has been supported by the ANR projects {StoQ} (grant number ANR-14-CE25-0003-01) and {NEXT} (grant number ANR-10-LABX-0037-NEXT), and by the PHC Sakura program (grant number 38615VA). This project has received funding from the European Research Council (ERC) under the European Union's Horizon 2020 research and innovation programme (grant agreement No 648913). The authors would like to thank the Institut Henri Poincar{\'e} in Paris for its hospitality and for hosting the trimester on ``Analysis in Quantum Information Theory'', during which part of this work was undertaken.

\bibliographystyle{alpha}
\bibliography{ref}
\end{document}